\documentclass[oneside, 11pt, reqno]{amsart}
\usepackage[left=1in, right=1in, top=1in, bottom=1in]{geometry}

\usepackage[usenames,dvipsnames]{xcolor}
\usepackage{amsmath}
\DeclareMathOperator*{\argmax}{argmax}
\usepackage{graphicx,enumitem}
\usepackage{mathtools}
\usepackage{amssymb} 
\usepackage{soul}
\usepackage[authoryear,round]{natbib}
\setcitestyle{authoryear,round}

\usepackage{booktabs, multirow}

\usepackage[margin=1in]{caption}
\allowdisplaybreaks

\usepackage[margin=0.7cm]{caption}
\usepackage{multirow}

\newcommand{\Zitkovic}[1]{{\v Z}itkovi\'c}
\newcommand{\Sirbu}[1]{S\^\i rbu}

\newcommand\tv{{\tilde{v}}}

\numberwithin{equation}{section}
\theoremstyle{plain}                
\newtheorem{theorem}{Theorem}

\theoremstyle{definition}           
\newtheorem{definition}[theorem]{Definition}
\newtheorem{proposition}[theorem]{Proposition}

\theoremstyle{remark}

\newcommand{\thmref}[1]{Theorem~\ref{#1}}
\newcommand{\proref}[1]{Proposition~\ref{#1}}

\usepackage{graphicx} 
\usepackage{xcolor}
\usepackage{setspace}

\newcommand\bv{{\bar{v}}}

\usepackage{hyperref}

\begin{document}

\author{Seongjin Kim$^\dag$ and Jin Hyuk Choi$^\ddag$}

\thanks{
$^\dag$Sungkyunkwan University (seongjinkim@skku.edu).\\
$^\ddag$Ulsan National Institute of Science and Technology (jchoi@unist.ac.kr).
}

\title{Mandatory Disclosure in Oligopolistic Market Making}

\begin{abstract} 
We develop a multi-period Kyle-type model that incorporates both mandatory disclosure of informed trades and imperfect competition among market makers. We prove the existence and uniqueness of a linear equilibrium and show that the liquidity-enhancing effect of disclosure is fundamentally linked to the degree of market-making competition. Disclosure lowers trading costs by reducing price impact, and its marginal benefit is strictly larger when competition is weak. We empirically validate this prediction using the 2002 Sarbanes--Oxley Act disclosure reform as a natural experiment. A difference-in-differences analysis of U.S. equities confirms that the spread reduction following enhanced disclosure is significantly larger for stocks with fewer active market makers.
\end{abstract}
\maketitle



\section{Introduction}

Since the seminal work of \cite{K1985}, Kyle-type models have provided the standard framework for analyzing how information is incorporated into prices through strategic trading. This framework is widely used to evaluate the effects of mandatory disclosure, a fundamental component of financial regulation whose actual impact on market quality remains unclear. As emphasized by \cite{LW2016} and \cite{CHL2016}, the capital-market effects of transparency reforms are highly heterogeneous and contingent on institutional settings, often making empirical results difficult to interpret without a unifying theoretical lens.

Various studies have extended the Kyle framework to analyze disclosure, as seen in \cite{HHL2001}, \cite{Z2004}, \cite{GL2012}, and \cite{AMTY2015}, yet these studies rely on the assumption of perfectly competitive market makers. Conversely, a separate line of research, including \cite{B2001}, \cite{IN2013}, and \cite{Choi2025}, investigates the implications of imperfectly competitive market makers but omits the role of disclosure. We bridge these two strands by developing a model that incorporates both mandatory disclosure and oligopolistic market-making competition. Our framework builds on the disclosure mechanism of \cite{HHL2001} and the market-making structure in the spirit of \cite{Choi2025}. We prove the existence and uniqueness of a linear equilibrium and characterize how disclosure interacts with the degree of competition to shape key market outcomes, such as price impact and transaction price autocorrelation. In particular, while both literatures independently recognize the roles of disclosure and competition in shaping liquidity, neither addresses whether the benefit of disclosure varies systematically with the degree of competition. We show that it does.

Our central finding is that the liquidity-enhancing effect of disclosure is fundamentally linked to the intensity of market-making competition. While disclosure consistently reduces price impact, we show that its marginal benefit is strictly larger when competition among market makers is weak. Although this benefit diminishes as the market approaches perfect competition, it remains an important driver of liquidity even in highly competitive environments.

We test this prediction empirically using the enactment of Section~403 of the Sarbanes--Oxley Act (SOX) on August~29, 2002, which abruptly shortened insider-trading disclosure deadlines from ten calendar days to two business days. Employing a difference-in-differences design with daily stock-level data from CRSP, we find that the post-SOX reduction in bid--ask spreads is significantly larger for stocks covered by fewer market makers, with a positive and highly significant interaction coefficient. Cross-country evidence from \cite{FGH2006} and \cite{LLM2012}, discussed in the Online Appendix, provides complementary support.

\section{Model and Equilibrium}

We consider a discrete-time trading model over the finite horizon $[0,1]$.
Trading occurs at equally spaced dates $0 = t_0 < t_1 < \cdots < t_N = 1$, with $\Delta t = 1/N$ and $t_n = n/N$, resulting in $N$ trading rounds. A single risky asset is traded; its terminal value $\tv$ is normally distributed as $\mathcal N(0,\sigma_v^2)$ with $\sigma_v^2>0$.
The realization of $\tilde v$ becomes common knowledge only after the final trading round at $t_N=1$.
There are three types of market participants in the economy:
\begin{itemize}
    \item \textbf{Noise traders.}
    Noise traders submit non-strategic market orders.
    Let $\Delta u_n$ denote the aggregate noise order at trading round $n$.
    The sequence $\{\Delta u_n\}_{n=1}^N$ is assumed to be i.i.d. and normally distributed, $\Delta u_n \sim \mathcal N(0,\sigma_u^2\Delta t)$, with $\sigma_u^2>0$.
    These orders are independent of the asset fundamental $\tilde v$.

    \item \textbf{Informed trader.}
    The informed trader observes the realization of $\tilde v$ at $t_0=0$ and retains this information throughout the trading horizon.
    Let $x_n$ be the cumulative order submitted by the informed trader up to trading round $n$, and $\Delta x_n := x_n - x_{n-1}$ the market order submitted at round $n$.
    
    We assume a disclosure regime where these informed orders ($\Delta x_n$) are publicly revealed after each trading round. Consequently, the history of informed trades $\{\Delta x_1,\ldots,\Delta x_{n-1}\}$ is common knowledge at time $t_n$.
    As shown in \cite{HHL2001}, invertible trading strategies under such disclosure would lead to full revelation of private information and infinite profits; thus, they cannot arise in equilibrium. 
    To prevent full revelation, the informed trader adds a noise component to the order, a mechanism known as dissimulation.
    Specifically, the informed trader submits mixed strategies of the form
    \[
    \Delta x_n = \text{(informational component)} + \Delta z_n,
    \]
    where $\{\Delta z_n\}_{n=1}^N$ is a sequence of independent Gaussian random variables satisfying $\Delta z_n \sim \mathcal N(0,\sigma_{z,n}^2\Delta t)$.
    The dissimulation noise $\Delta z_n$ is independent of both $\tilde v$ and noise traders orders $\{\Delta u_n\}_{n=1}^N$.

    \item \textbf{Market makers.}
    There are $M \ge 3$ risk-neutral market makers, indexed by $m=1,\ldots,M$, who provide liquidity via price-dependent orders. Let $\Delta y_n^m(p_n)$ denote the order submitted by market maker $m$ at round $n$ as a function of the transaction price $p_n$.

    Under the disclosure regime, market makers condition their orders on the clearing price and the history of disclosed informed trades. Accordingly, $\Delta y_n^m$ is $\sigma(p_n,\Delta x_1,\ldots,\Delta x_{n-1})$-measurable. 
\end{itemize}

The aggregate order flow at trading round $n$ is defined as
\begin{equation}\label{aggregate_order}
    \Delta w_n := \Delta x_n + \Delta u_n.
\end{equation}
The transaction price $p_n$ is determined by the market clearing condition:
\begin{equation}\label{market_clearing_condition}
\sum_{m=1}^M \Delta y_n^m(p_n) + \Delta w_n= 0.
\end{equation}

\begin{definition}
An equilibrium is a profile of trading strategies $\left( \{\Delta \hat x_n\}_{n}, \{\Delta \hat y_n^m\}_{n,m} \right)$, together with a price process $\{p_n\}_{n}$ satisfying the market clearing condition \eqref{market_clearing_condition}, such that:
\begin{enumerate}
    \item
    Given the market makers' strategies $\{\Delta \hat y_n^m\}_{n}$,
    the informed trader maximizes expected profit conditional on the trader's information:
    \[
    \{\Delta \hat x_n\}_{n}
    \in
    \argmax_{\{\Delta x_n\}_{n}}
    \mathbb E\!\left[
        \sum_{n=1}^N (\tilde v - p_n)\,\Delta x_n
        \,\bigg|\,
        \tilde v
    \right].
    \]

    \item For each $m\in \{1,\ldots,M\}$, given the informed trader's strategy $\{\Delta \hat x_n\}_{n}$
    and the strategies of other market makers
    $\{\Delta \hat y_n^\ell\}_{n}$ for $\ell \neq m$,
    market maker $m$ maximizes expected profit:
    \[
    \{\Delta \hat y_n^m\}_{n}
    \in
    \argmax_{\{\Delta y_n^m\}_{n}}
    \mathbb E\!\left[
        \sum_{n=1}^N (\tilde v - p_n)\,\Delta y_n^m
    \right].
    \]
\end{enumerate}
\end{definition}

We focus on linear equilibria, in which all agents' trading strategies are affine functions of the information available to them.
To characterize the information structure induced by disclosure, we define the filtration
\[
\mathcal F_n := \sigma(\Delta x_1,\ldots,\Delta x_n),
\]
which is generated by the history of disclosed informed orders up to round $n$. Starting from the prior $\bar v_0 = \mathbb E[\tilde v] = 0$, we define the post-disclosure belief at round $n$ as
\begin{equation}\label{vbar}
    \bar v_n:= \mathbb{E}[\tilde v|\mathcal{F}_n].
\end{equation}
This belief represents the market's updated estimate of the asset value after the $n$-th informed order $\Delta x_n$ is disclosed.

We conjecture that the informed trader's order and the market makers' orders take linear forms consistent with this filtration.
Specifically, for each round $n$, the informed trader submits
\begin{equation}\label{disc_informed_strategy}
\Delta \hat x_n=
\beta_n(\tilde v-\bar v_{n-1})\,\Delta t
+
s_n
+
\Delta z_n,
\end{equation}
where $\beta_n$ is a deterministic coefficient, $\bar v_{n-1}$ is the belief defined in \eqref{vbar}, $s_n$ is $\mathcal F_{n-1}$-measurable, and $\Delta z_n$ is the dissimulation noise. The informed trader chooses the variance $\sigma_{z,n}^2$ of $\Delta z_n$.
At the same trading round, each market maker $m$ submits a price-dependent order
\footnote{Note that we do not impose symmetry among market makers as an ex-ante assumption. We allow the $M$ market makers to adopt heterogeneous pricing and inventory strategies, and show in Theorem~\ref{thm:equilibrium} that the unique linear equilibrium is necessarily symmetric.}
\begin{equation}\label{disc_mm_strategy}
    \Delta \hat y_n^m = \gamma^m_n(\bar v_{n-1} - p_n)+ r^m_n,
\end{equation}
where $\gamma_n^m$ is a deterministic coefficient and $r^m_n$ is $\mathcal F_{n-1}$-measurable.

Given the linear structure and the Gaussian assumptions, it follows that $\tv - \bv_{n-1}$ is independent of $\mathcal{F}_{n-1}$. Consequently, the Gaussian projection theorem yields
\begin{equation}\label{belief_update_mean}
\Delta \bar v_n
:= \bv_n - \bv_{n-1}=
\mathbb E\!\left[
\tilde v-\bar v_{n-1}
\mid
\mathcal F_n
\right]
=
\mathbb E\!\left[
\tilde v-\bar v_{n-1}
\mid
\Delta x_n - s_n
\right]
=
\psi_n(\Delta x_n - s_n),
\end{equation}
where the projection coefficient $\psi_n$ and the deterministic conditional variance $\Sigma_n := \mathbb E[(\tilde v - \bar v_n)^2]$ satisfy
\begin{align}
\psi_n
=
\frac{\beta_n\Sigma_{n-1}}
{\beta_n^2\Sigma_{n-1}\Delta t+\sigma_{z,n}^2}, \quad
\Sigma_n
=
\Sigma_{n-1}
-
\psi_n^2\bigl(\beta_n^2\Sigma_{n-1}\Delta t+\sigma_{z,n}^2\bigr)\Delta t .\label{variance_recursion}
\end{align}

In addition to the post-disclosure belief $\bar v_n$, we consider the pre-disclosure belief 
\begin{align*}
d_n := \mathbb E[\tilde v \mid \mathcal F_{n-1} \vee \sigma(\Delta w_n)],
\end{align*} 
which incorporates the aggregate order flow $\Delta w_n$ but precedes the disclosure of $\Delta x_n$. As before, the Gaussian projection theorem yields
\begin{equation}\label{phi_definition}
d_n = \bar v_{n-1} + \phi_n(\Delta w_n - s_n), \quad \text{where} \quad    \phi_n := \frac{\beta_n\Sigma_{n-1}} {\beta_n^2\Sigma_{n-1}\Delta t+\sigma_{z,n}^2 +\sigma_u^2}.
\end{equation}

Given the market clearing condition \eqref{market_clearing_condition}, the transaction price $p_n$ is determined endogenously. In their respective optimization problems, the informed trader perceives the price as
\begin{equation}\label{perceived_price_informed}
p_n = \bar v_{n-1} + \frac{1}{\gamma_n^\Sigma} \left( \Delta x_n + \Delta u_n + r_n^\Sigma \right),
\end{equation}
while each market maker $m$ perceives the price as
\begin{equation}\label{perceived_price_mm}
p_n = \bar v_{n-1} + \frac{1}{\gamma^\Sigma_n - \gamma^m_n} \left( \Delta y^m_n + \Delta \hat x_n + \Delta u_n + r^\Sigma_n - r^m_n \right).
\end{equation}
where $\gamma_n^\Sigma := \sum_{j=1}^M \gamma_n^j$ and $r_n^\Sigma := \sum_{j=1}^M r_n^j$ are the aggregate slope and intercept terms, respectively.

By imposing the optimality conditions from the definition of equilibrium and verifying the conjectured linear forms, we obtain the following characterization of the unique linear equilibrium.

\begin{theorem}\label{thm:equilibrium}
There exists a unique linear equilibrium, which is symmetric across market makers.
Specifically, there exist deterministic sequences
$\{\beta_n,\sigma_{z,n},\gamma_n,\psi_n\}_{n=1}^N$ such that, for each trading round $n$, the equilibrium trading strategies are given by
\begin{align*}
\Delta \hat x_n &= \beta_n(\tilde v-\bar v_{n-1})\,\Delta t + \Delta z_n,\\
\Delta \hat y_n^m &= \gamma_n(\bar v_{n-1}-p_n),\\
\Delta \bar v_n &= \psi_n \Delta \hat x_n,
\end{align*}
where the coefficients are 
\begin{align*}
    \beta_n &= \frac{N}{N-n+1}\sqrt{\frac{M-2}{M}}\frac{\sigma_u}{\sigma_v},   &&\sigma_{z,n}^2 = \frac{N-n}{N-n+1}\frac{M-2}{M}\sigma_u^2,\\
    \gamma_n &= \frac{2\sqrt{M-2}}{M\sqrt{M}}\frac{\sigma_u}{\sigma_v}, &&\psi_n = \sqrt{\frac{M}{M-2}}\frac{\sigma_v}{\sigma_u}.
\end{align*}
The associated value functions for the informed trader and each market maker $m$ are
\begin{align*}
&\max_{\{\Delta x_k\}_{k=n+1}^N}
\mathbb E\Big[
\sum_{k=n+1}^N(\tilde v-p_k)\Delta x_k
\;\Big|\;
\mathcal F_n\vee\sigma(\tilde v)
\Big]
= \frac{1}{2}\sqrt{\frac{M-2}{M}}\frac{\sigma_u}{\sigma_v}(\tilde v-\bar v_n)^2,\\
&\max_{\{\Delta y_k^m\}_{k=n+1}^N}
\mathbb E\Big[
\sum_{k=n+1}^N(\tilde v-p_k)\Delta y_k^m
\;\Big|\;
\mathcal F_{n}
\Big]
= \frac{N-n}{N} \frac{1}{M\sqrt{M(M-2)}} \sigma_u \sigma_v.
\end{align*}
\end{theorem}
\begin{proof}
See the Online Appendix for the proof.
\end{proof}

Several remarks on the equilibrium properties are in order. First, the existence of a symmetric linear equilibrium requires $M > 2$. Second, several features of our equilibrium are consistent with the perfect competition benchmark of \cite{HHL2001}. Specifically, the intercept terms $s_n$ in \eqref{disc_informed_strategy} and $r_n^m$ in \eqref{disc_mm_strategy} vanish, and both the price slope $\gamma_n$ and disclosure sensitivity $\psi_n$ remain constant over time.
Furthermore, the informed trader strategically determines the dissimulation noise variance $\sigma_{z,n}^2$ to regulate the pace of information revelation, with this noise term vanishing only at the terminal round $n=N$. Third, as the number of market makers $M$ approaches infinity, our equilibrium coefficients converge to the perfectly competitive results of \cite{HHL2001}. However, a key distinction in our model is that imperfect competition among market makers allows them to earn strictly positive expected profits, in contrast to the zero-profit outcome in the perfect competition setting.

\section{Comparative Analysis}

We evaluate the economic implications of our model by contrasting it with three benchmark Kyle-type environments: the perfect-competition model without disclosure (continuous time limit in \cite{K1985}), the perfect-competition model with disclosure (\cite{HHL2001}), and the imperfect-competition model without disclosure (\cite{Choi2025}). Table~\ref{tab:model_comparison} summarizes the equilibrium price impact, agents' expected profits, and price autocorrelation across these market structures.

\begin{table}[t]
\centering
\caption{Comparison of Equilibrium Outcomes across Market Structures}
\label{tab:model_comparison}
\begin{tabular}{lccccc}
\toprule
 & \multicolumn{2}{c}{\textbf{Perfect Competition}} & & \multicolumn{2}{c}{\textbf{Imperfect Competition}} \\
\cmidrule{2-3} \cmidrule{5-6}
 & No Disclosure & Disclosure & & No Disclosure & Disclosure \\
 & {\scriptsize Limit of \cite{K1985}} & {\scriptsize \cite{HHL2001}} & & {\scriptsize \cite{Choi2025}} &  {\scriptsize (This study)}  \\
\midrule
Price impact & $\frac{\sigma_v}{\sigma_u}$ & $\frac{1}{2} \frac{\sigma_v}{\sigma_u}$ & & $\frac{M-1}{M-2} \frac{\sigma_v}{\sigma_u}$ & $\frac{\sqrt{M}}{2\sqrt{M-2}} \frac{\sigma_v}{\sigma_u}$ \\
\addlinespace
Informed trader's profit & $\sigma_v \sigma_u$ & $\frac{1}{2}  \sigma_v \sigma_u$ & & $\sigma_v \sigma_u$ & $\frac{\sqrt{M-2}}{2\sqrt{M}}\sigma_v \sigma_u$  \\
\addlinespace
Market makers' profit      & 0 & 0 & & $\frac{1}{M-2}  \sigma_v \sigma_u$ & $\frac{1}{\sqrt{M(M-2)}}  \sigma_v \sigma_u$ \\
\addlinespace
Noise traders' profit      & $-\sigma_v \sigma_u$  & $-\frac{1}{2}  \sigma_v \sigma_u$ & &  $-\frac{M-1}{M-2}  \sigma_v \sigma_u$ & $-\frac{\sqrt{M}}{2\sqrt{M-2}} \sigma_v \sigma_u$ \\
\addlinespace
Price autocorrelation
& $0$ 
& $0$ 
& 
& $-\frac{M-1}{(M-1)^2+1}$ 
& $-\frac{1}{2(M-1)}$ \\
\bottomrule
\end{tabular}

\vspace{0.5em}
\parbox{\textwidth}
{\scriptsize
\textit{Note}: This table compares equilibrium price impact, expected profits, and price autocorrelation across different market settings. $\sigma_v$ and $\sigma_u$ denote the fundamental and noise volatility, respectively, and $M$ is the number of market makers.}
\end{table}

Price impact refers to the linear coefficient relating transaction prices to aggregate order flow, corresponding to the Kyle $\lambda$; in the imperfectly competitive environments considered here and in \cite{Choi2025}, this object is given by $\frac{1}{M\gamma}$. The profits for the informed trader and market makers represent ex-ante expectations over the entire trading horizon, 
\begin{align*}
\Pi_I := \mathbb{E}\Big[\sum_{k=1}^N(\tilde v - p_k)\Delta x_k\Big] \quad \textrm{and} \quad \Pi_M := \mathbb{E}\Big[\sum_{k=1}^N(p_k - \tilde v)\Delta w_k\Big],
\end{align*}
while the noise traders' profit is defined residually as 
$$\Pi_N := -\Pi_I - \Pi_M.$$ Price autocorrelation is measured as the limit 
\begin{equation*}
    \lim_{N\to\infty}\sup_{n= 1, \dots, N}\mathrm{Corr}(\Delta p_n,\Delta p_{n+1}).
\end{equation*}
For our model (imperfect competition with disclosure), the detailed derivation of the price autocorrelation is provided in the Online Appendix.

The comparison in Table~\ref{tab:model_comparison} leads to the following key observations regarding how disclosure interacts with market competition.

\begin{proposition}\label{cor:disclosure_impact}
Compared to the perfect competition ($M=\infty$) benchmark, the introduction of mandatory disclosure under imperfect competition ($3\leq M < \infty$) yields the following relative effects:
\begin{enumerate}
\item \textbf{Greater Liquidity Improvement}: The proportional reduction in price impact induced by disclosure is strictly larger under imperfect competition.
\begin{align*}
\Big(\tfrac{\sigma_v}{\sigma_u}- \tfrac{1}{2} \tfrac{\sigma_v}{\sigma_u}\Big)/ \Big(\tfrac{\sigma_v}{\sigma_u}\Big) < \Big(\tfrac{M-1}{M-2} \tfrac{\sigma_v}{\sigma_u}- \tfrac{\sqrt{M}}{2\sqrt{M-2}} \tfrac{\sigma_v}{\sigma_u}\Big)/\Big(\tfrac{M-1}{M-2} \tfrac{\sigma_v}{\sigma_u}\Big).
\end{align*}

\item \textbf{Reduction in Informational and Oligopolistic Rents}: 
Under imperfect competition, mandatory disclosure leads to a larger reduction in the informed trader's expected profit than under perfect competition,
\[
\Big(\sigma_v \sigma_u - \tfrac{1}{2}\sigma_v \sigma_u\Big)
<
\Big(\sigma_v \sigma_u - \tfrac{\sqrt{M-2}}{2\sqrt{M}}\sigma_v \sigma_u\Big).
\]
Moreover, disclosure compresses market makers' oligopolistic profits,
\[
\tfrac{1}{\sqrt{M(M-2)}}\sigma_v \sigma_u
<
\tfrac{1}{M-2}\sigma_v \sigma_u.
\]


\item \textbf{Greater Reduction in Noise Traders' Costs}: The decrease in expected trading costs for noise traders is more pronounced when market-making competition is weaker.
\begin{align*}
\Big(\sigma_v \sigma_u- \tfrac{1}{2} \sigma_v \sigma_u\Big) < \Big(\tfrac{M-1}{M-2} \sigma_v \sigma_u- \tfrac{\sqrt{M}}{2\sqrt{M-2}} \sigma_v \sigma_u\Big).
\end{align*}

\item \textbf{Improvement in Price Efficiency}: Disclosure moves the negative price autocorrelation induced by imperfect competition closer to zero.
\begin{align*}
-\tfrac{M-1}{(M-1)^2+1}<-\tfrac{1}{2(M-1)}<0.
\end{align*}

\end{enumerate}
\end{proposition}

These results suggest that disclosure serves as an effective mechanism to mitigate market frictions, particularly when market-making competition is weak.
Furthermore, while imperfect competition inherently induces negative price autocorrelation, as shown in \cite{Choi2025}, the introduction of mandatory disclosure moves this autocorrelation closer to zero.

In the next section, we provide direct empirical evidence supporting these theoretical predictions.


\section{Empirical Evidence}\label{sec:empirical}

In this section, we provide direct empirical evidence for the central prediction of \proref{cor:disclosure_impact}: the liquidity-enhancing effect of mandatory disclosure is strictly larger when competition among market makers is weaker. We exploit the enactment of Section~403 of the Sarbanes--Oxley Act (SOX) as a natural experiment and test whether the resulting improvement in bid--ask spreads varies with the number of active market makers.

\subsection{Testable prediction}

\proref{cor:disclosure_impact} establishes that, under imperfect competition, the proportional reduction in price impact induced by disclosure is strictly larger when the number of market makers $M$ is small. Because the bid--ask spread is a standard empirical proxy for the theoretical price impact (Kyle's $\lambda$), we formulate the following testable prediction.

\begin{quote}
\textit{Following a strengthening of mandatory disclosure, the reduction in bid--ask spreads is larger for stocks with fewer market makers.}
\end{quote}

\subsection{Data and institutional setting}

\paragraph{Event.}
On August~29, 2002, the SEC adopted accelerated insider-trading disclosure rules under SOX Section~403, shortening the filing deadline from ten calendar days to two business days after the transaction date. This regulatory change constitutes a significant and sudden strengthening of mandatory disclosure obligations for corporate insiders, providing a clean natural experiment in which the pre-event period (with its significant reporting delays) proxies for the theoretical no-disclosure benchmark in our model, while the post-event period reflects the enhanced disclosure regime.

\paragraph{Sample.}
We obtain daily stock-level data from the Center for Research in Security Prices (CRSP) Daily Stock File via WRDS. Our sample covers the period March~1 through December~31, 2002, centered around the August~29 event date. We restrict the sample to common stocks (SHRCD $\in \{10,11\}$) listed on the NYSE or NASDAQ (EXCHCD $\in \{1,3\}$). After imposing standard filters---requiring valid closing prices, bid and ask quotes, positive trading volume, and a minimum of 15 prior trading days for volatility estimation---the final sample comprises 692,453 stock-day observations across 4,046 unique stocks (466 NYSE; 3,580 NASDAQ).

A key variable in our analysis is MMCNT, the CRSP-reported count of market makers (dealers) posting quotes for a given stock on a given day. This variable captures all liquidity providers submitting quotes, including specialists, floor brokers, and supplemental liquidity providers on the NYSE, and competing dealers on the NASDAQ. Because MMCNT exhibits day-to-day fluctuation driven by transient factors, we compute the time-series average $\overline{\text{MMCNT}}_i$ for each stock $i$ over the full sample period and treat it as a stock-level constant. This averaging procedure isolates the cross-sectional variation in the intensity of market-making competition, which is the relevant dimension for testing \proref{cor:disclosure_impact}. The distribution of $\overline{\text{MMCNT}}_i$ differs markedly across exchanges: NYSE-listed stocks have a median of approximately 3 (mean $\approx$ 7), while NASDAQ-listed stocks have a median of approximately 18 (mean $\approx$ 23). This approximately three-fold difference provides substantial cross-sectional variation for identifying the heterogeneous effect of disclosure.

Table~\ref{tab:sumstats} reports summary statistics for the key variables, partitioned by event period and exchange.

\begin{table}[t]
\centering
\caption{Summary statistics}
\label{tab:sumstats}
\begin{tabular}{lrrrrrr}
\toprule
 & \multicolumn{3}{c}{\textbf{Pre-SOX}} & \multicolumn{3}{c}{\textbf{Post-SOX}} \\
\cmidrule(lr){2-4} \cmidrule(lr){5-7}
Variable & Mean & Median & Std & Mean & Median & Std \\
\midrule
\multicolumn{7}{l}{\textit{Panel A: Full sample} ($N_{\text{pre}} = 398{,}321$; $N_{\text{post}} = 294{,}132$)} \\
\addlinespace
Spread (\%) & 2.50 & 1.12 & 4.02 & 2.50 & 1.09 & 4.08 \\
$\overline{\text{MMCNT}}$ & 20.7 & 17.1 & 14.9 & 20.5 & 17.0 & 14.6 \\
MktCap (\$M) & 987 & 124 & 7{,}188 & 840 & 102 & 6{,}326 \\
Volume (shares) & 544{,}756 & 44{,}525 & 5{,}464{,}530 & 463{,}159 & 37{,}109 & 2{,}867{,}266 \\
Volatility (20d) & 0.046 & 0.037 & 0.036 & 0.050 & 0.040 & 0.041 \\
\addlinespace
\multicolumn{7}{l}{\textit{Panel B: NYSE} ($n = 466$ stocks)} \\
\addlinespace
Spread (\%) & 1.17 & 0.70 & 1.70 & 1.18 & 0.68 & 1.84 \\
$\overline{\text{MMCNT}}$ & 6.7 & 3.0 & 7.8 & 6.6 & 3.0 & 7.8 \\
\addlinespace
\multicolumn{7}{l}{\textit{Panel C: NASDAQ} ($n = 3{,}580$ stocks)} \\
\addlinespace
Spread (\%) & 2.69 & 1.24 & 4.21 & 2.69 & 1.21 & 4.28 \\
$\overline{\text{MMCNT}}$ & 22.7 & 18.5 & 14.6 & 22.5 & 18.5 & 14.2 \\
\bottomrule
\end{tabular}

\vspace{0.5em}
\parbox{\textwidth}
{\scriptsize
\textit{Note}: The sample consists of NYSE and NASDAQ common stocks from the CRSP Daily Stock File, March--December 2002. Pre-SOX denotes the period before August~29, 2002; Post-SOX denotes the period on or after that date. Spread is defined as $(\text{Ask} - \text{Bid})/\text{Midpoint} \times 100$. $\overline{\text{MMCNT}}$ is the stock-level time-series average of the daily market maker count. Volatility is the rolling 20-trading-day standard deviation of daily returns. MktCap is price times shares outstanding.}
\end{table}

\subsection{Empirical model}

We estimate the following difference-in-differences regression:
\begin{equation}\label{eq:did}
\begin{split}
\log(\text{Spread}_{i,t}) &= \alpha + \beta_1 \, \text{Post}_t + \beta_2 \, \log(1+\overline{\text{MMCNT}}_i) \\
&\qquad + \beta_3 \, \text{Post}_t \times \log(1+\overline{\text{MMCNT}}_i) + \boldsymbol{\gamma}' \mathbf{X}_{i,t} + \varepsilon_{i,t},
\end{split}
\end{equation}
where $\text{Spread}_{i,t} := (\text{Ask}_{i,t} - \text{Bid}_{i,t}) / \text{Midpoint}_{i,t}$ is the relative bid--ask spread for stock $i$ on day $t$, and $\text{Post}_t$ is a dummy variable equal to one for dates on or after August~29, 2002. The variable $\log(1+\overline{\text{MMCNT}}_i)$ captures the cross-sectional variation in market-making competition; the logarithmic transformation reflects the theoretical prediction that the marginal effect of an additional market maker diminishes as $M$ increases. The control vector $\mathbf{X}_{i,t}$ includes $\log(\text{MktCap}_{i,t})$, $\log(\text{Volume}_{i,t})$, the 20-day rolling return volatility, and $1/\text{Price}_{i,t}$ (to absorb the mechanical tick-size effect on percentage spreads).

The coefficient of primary interest is $\beta_3$, which captures the differential effect of disclosure across stocks with varying levels of market-making competition. \proref{cor:disclosure_impact} predicts $\beta_3 > 0$: since higher $\overline{\text{MMCNT}}$ corresponds to stronger competition, the spread reduction following disclosure should be smaller for such stocks, implying a positive interaction term in a log-spread regression. The model also predicts $\beta_1 < 0$ (disclosure reduces spreads on average) and $\beta_2 < 0$ (more competition is associated with lower spreads in the cross section).

Standard errors are clustered at the stock level to account for within-stock serial correlation in the error term.

\subsection{Results}

Table~\ref{tab:regression} reports the estimation results. All three coefficient predictions are confirmed with high statistical significance.

\begin{table}[t]
\centering
\caption{Difference-in-differences regression: log(Spread) on disclosure and competition}
\label{tab:regression}
\begin{tabular}{lcccc}
\toprule
& Coefficient & Std.\ Error & $t$-statistic & $p$-value \\
\midrule
Intercept & $1.929$ & $0.091$ & $21.25$ & ${<}\,0.001$ \\
$\text{Post}$ ($\beta_1$) & $-0.374$ & $0.020$ & $-18.75$ & ${<}\,0.001$ \\
$\log(1+\overline{\text{MMCNT}})$ ($\beta_2$) & $-0.378$ & $0.014$ & $-26.85$ & ${<}\,0.001$ \\
$\text{Post} \times \log(1+\overline{\text{MMCNT}})$ ($\beta_3$) & $0.105$ & $0.007$ & $15.73$ & ${<}\,0.001$ \\
\addlinespace
$\log(\text{MktCap})$ & $-0.449$ & $0.010$ & $-47.49$ & ${<}\,0.001$ \\
$\log(\text{Volume})$ & $-0.047$ & $0.005$ & $-9.86$ & ${<}\,0.001$ \\
Volatility (20d) & $4.989$ & $0.280$ & $17.85$ & ${<}\,0.001$ \\
$1/\text{Price}$ & $0.086$ & $0.010$ & $8.93$ & ${<}\,0.001$ \\
\midrule
Observations & \multicolumn{4}{c}{692,453} \\
Clusters (stocks) & \multicolumn{4}{c}{4,046} \\
$R^2$ & \multicolumn{4}{c}{0.533} \\
\bottomrule
\end{tabular}

\vspace{0.5em}
\parbox{\textwidth}
{\scriptsize
\textit{Note}: The dependent variable is $\log(\text{Spread}_{i,t})$, where Spread is the relative bid--ask spread. $\text{Post}_t = 1$ for dates $\geq$ August~29, 2002. $\overline{\text{MMCNT}}_i$ is the stock-level average market maker count. Standard errors are clustered by stock (PERMNO). All coefficients are significant at the 1\% level.}
\end{table}

The estimate $\hat\beta_1 = -0.374$ ($t = -18.75$) indicates that, on average, bid--ask spreads declined significantly following the implementation of SOX Section~403. This is consistent with the prediction of disclosure-only models such as \cite{HHL2001}, where mandatory disclosure reduces information asymmetry and price impact.

The estimate $\hat\beta_2 = -0.378$ ($t = -26.85$) confirms that stocks with more market makers exhibit lower spreads in the cross section. This cross-sectional pattern is consistent with models that isolate the competition channel, such as \cite{Choi2025}, where price impact is a decreasing function of the number of market makers.

The key estimate is $\hat\beta_3 = 0.105$ ($t = 15.73$), which is positive and highly significant. This indicates that the disclosure-induced reduction in spreads is attenuated for stocks with a greater number of market makers. Equivalently, stocks with fewer market makers---corresponding to weaker competition---experienced a larger improvement in liquidity following the regulatory change. This finding directly supports the prediction of \proref{cor:disclosure_impact}.

To assess economic significance, we evaluate the estimated total effect of SOX on log-spreads at representative levels of $\overline{\text{MMCNT}}$:
\[
\hat\beta_1 + \hat\beta_3 \times \log(1 + \overline{\text{MMCNT}}).
\]
For a typical NYSE stock with $\overline{\text{MMCNT}} \approx 3$ (sample median), the estimated effect is $-0.374 + 0.105 \times \log 4 \approx -0.229$, which corresponds to a 20.5\% reduction in spreads ($\exp(-0.229) - 1$). For a typical NASDAQ stock with $\overline{\text{MMCNT}} \approx 18$ (sample median), the estimated reduction is $-0.374 + 0.105 \times \log 19 \approx -0.065$, or a 6.3\% reduction. The difference of 14.2 percentage points demonstrates that the economic magnitude of the heterogeneous disclosure effect is substantial.

Figure~\ref{fig:sox_effect} provides a graphical illustration of this heterogeneous effect by plotting the estimated spread reduction against $\overline{\text{MMCNT}}$.

\begin{figure}[t]
\centering
\includegraphics[width=0.5\linewidth]{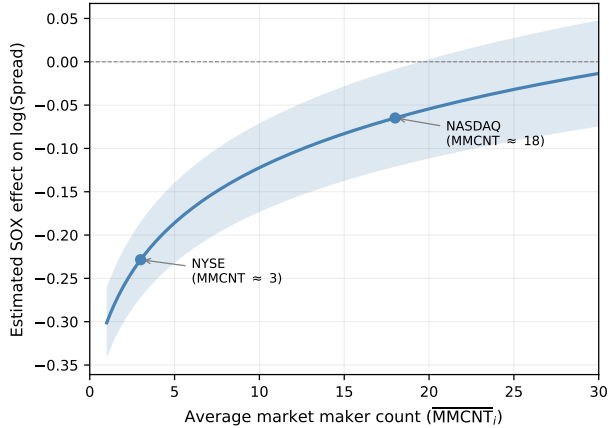}
\caption{Estimated effect of SOX Section~403 on log(Spread) as a function of market maker count. The solid curve plots $\hat\beta_1 + \hat\beta_3 \times \log(1+\overline{\text{MMCNT}})$ using estimates from Table~\ref{tab:regression}. The shaded region indicates the 95\% confidence band. Stocks with fewer market makers experience a larger reduction in spreads following the disclosure reform, consistent with \proref{cor:disclosure_impact}.}
\label{fig:sox_effect}
\end{figure}

\subsection{Complementary cross-country evidence}

The analysis above exploits within-country variation across U.S.\ stocks.
Complementary evidence can be drawn from the cross-country literature.
\cite{FGH2006} document a positive association between disclosure strength
and market development across 50 international exchanges.
\cite{LLM2012} show that the transparency--liquidity relation is two to
three times stronger in countries with weak governance institutions than in
countries with strong ones.
To the extent that less developed markets and weaker governance
environments are associated with fewer active intermediaries, both sets of
findings are consistent with \proref{cor:disclosure_impact}'s prediction that the marginal
benefit of disclosure is amplified when competition among market makers is
limited. A detailed reinterpretation appears in the Online Appendix.

\section{Conclusion}

This paper develops a Kyle-type model that explicitly incorporates imperfect competition among market makers into the analysis of disclosure. By solving this model, we show that the liquidity-enhancing effect of disclosure is fundamentally contingent on the degree of market-making competition. Our analysis reveals that disclosure provides a greater marginal liquidity benefit in markets where competition is weak.
We empirically validate this central prediction using the 2002 Sarbanes--Oxley Act disclosure reform as a natural experiment, demonstrating that the post-SOX reduction in bid--ask spreads is significantly larger for U.S.\ equities with fewer active market makers.
Overall, these results highlight the importance of considering the underlying competitive structure of the financial sector when evaluating the effectiveness of disclosure requirements.

Our analysis assumes an exogenously fixed number of market makers and a single informed trader.
Two extensions merit discussion.
First, in the endogenous participation model of \citet{IN2013},
dealers enter until expected profit equals the entry cost,
so fewer dealers are active when entry costs are high.
Mandatory disclosure would lower adverse selection and raise dealer profits,
with the profit gain larger when fewer dealers share the reduced information rent.
The incentive for new entry is therefore strongest in low-competition markets,
which is where \proref{cor:disclosure_impact} predicts the largest disclosure benefit,
so endogenous participation \emph{reinforces} our main result.
Second, \citet{GL2012} show that competition among multiple informed traders under a disclosure requirement leads to more aggressive trading, which accelerates price discovery and improves liquidity.
When this multi-insider effect is combined with imperfect competition among market makers, the reinforcement is twofold: disclosure attracts more aggressive informed trading while simultaneously yielding a larger proportional reduction in price impact in markets with fewer dealers.
Both channels point in the same direction as \proref{cor:disclosure_impact}, suggesting that our central prediction is robust to these generalizations.
Formally extending the model along both dimensions is a promising direction for future research.

\bigskip\bigskip




\bibliographystyle{abbrvnat}
\bibliography{reference_disclosure}

\newpage

\begin{center}
{\bf   Online Appendix}
\end{center}

\appendix
\section{Proof of Theorem~\ref{thm:equilibrium}}\label{proof}

The proof proceeds by backward induction. We first note that at the terminal trading round $n=N$, the informed trader has no incentive to conceal information for future periods, leading to zero dissimulation noise, $\sigma_{z,N}^2 = 0$. This terminal condition yields $\Sigma_N = 0$.

Fix $n\in\{1,\ldots,N-1\}$ and suppose that, for all trading rounds $k=n+1,\ldots,N$, the equilibrium coefficients and value functions are given by
\begin{align}
s_k &= 0, \quad r_k^m = 0, 
\label{induction_r_k}\\
\gamma_k^m &= \frac{M-2}{M(M-1)}\frac{1}{\phi_k} =: \gamma_k, 
\label{induction_gamma_k}\\
\psi_k &= \frac{2M-2}{M-2}\phi_k, 
\label{induction_psi_k}\\
\Sigma_k &= \frac{N-k}{N-k+1}\Sigma_{k-1}, 
\label{induction_Sigma_k}\\
\beta_k &= \frac{N}{N-k+1}\frac{M-2}{2M-2}\frac{1}{\phi_k}, 
\label{induction_beta_k}\\
\phi_k &= \sqrt{\frac{M(M-2)}{(2M-2)^2}\frac{N}{N-k+1}\frac{\Sigma_{k-1}}{\sigma_u^2}},
\label{induction_phi_k}\\
\sigma_{z,k}^2 &= \frac{M-2}{M}\frac{N-k}{N-k+1}\sigma_u^2, 
\label{induction_sigmaz_k}\\
\max_{\{\Delta x_j\}_{j=k}^N} &\mathbb{E}\Big[\sum_{j=k}^N(\tilde v - p_j)\Delta x_j 
\Big|\mathcal{F}_{k-1}\vee \sigma(\tilde v)\Big]
= \frac{M\gamma_k}{4}(\tilde v- \bar{v}_{k-1})^2,
\label{induction_IT_value}\\
\max_{\{\Delta y^m_j\}_{j=k}^N} &\mathbb{E}\Big[\sum_{j=k}^N(\tilde v - p_j)\Delta y^m_j 
\Big| \mathcal{F}_{k-1}\Big]
= \nu_{k-1}=\nu_k + \frac{\beta_k\Sigma_{k-1}\Delta t}{M(M-2)},
\label{induction_MM_value}
\end{align}
where $\nu_N=0$. Note that these expressions are consistent with the terminal conditions at $k=N$, such as $\sigma_{z,N}^2=0$ and $\Sigma_N=0$.

{\bf Informed trader's optimization problem}: 
Suppose that each market maker $m=1,\dots,M$ follows the strategy \eqref{disc_mm_strategy}. 
Using \eqref{belief_update_mean} and \eqref{perceived_price_informed}, the informed trader's value at trading round $n$ can be written as
\begin{align}
&\mathbb{E}\Big[(\tilde v- p_n)\Delta x_n 
+ \sum_{k=n+1}^N(\tilde v - p_k)\Delta x_k
\;\big|\; \mathcal{F}_n\vee \sigma(\tilde v)\Big] \nonumber\\
&= \mathbb{E}\Big[(\tilde v- p_n)\Delta x_n 
+ \frac{M\gamma_{n+1}}{4}(\tilde v- \bar{v}_{n})^2
\;\big|\; \mathcal{F}_n\vee \sigma(\tilde v)\Big] \nonumber\\
&= \mathbb{E}\Big[\Big(\tilde v  - \bar v_{n-1}  
- \frac{1}{\gamma^\Sigma_n}\big(\Delta x_n+\Delta u_n+r^\Sigma_n\big)\Big)\Delta x_n 
+ \frac{M\gamma_{n+1}}{4}\big(\tilde v- \bar{v}_{n-1} - \psi_n(\Delta x_n - s_n)\big)^2
\;\big|\; \mathcal{F}_n\vee \sigma(\tilde v)\Big] \nonumber\\
&= \Big(\tilde v  - \bar v_{n-1}  
- \frac{1}{\gamma^\Sigma_n}(\Delta x_n+r^\Sigma_n)\Big)\Delta x_n 
+ \frac{M\gamma_{n+1}}{4}\big(\tilde v- \bar{v}_{n-1} - \psi_n(\Delta x_n - s_n)\big)^2 . \label{value_iterate}
\end{align}
The first-order condition with respect to $\Delta x_n$ is
\begin{align*}
\Big(-\frac{2}{\gamma_n^\Sigma}+ \frac{M\gamma_{n+1}\psi_n^2}{2}\Big)\Delta x_n
+\Big(1-\frac{M\gamma_{n+1}\psi_n}{2}\Big)(\tilde v  - \bar v_{n-1})
-\frac{r^\Sigma_n}{\gamma^\Sigma_n}
-\frac{M\gamma_{n+1}\psi_n^2 s_n}{2}
=0 .
\end{align*}
For the mixed strategy with dissimulation noise to be optimal, the informed trader must be indifferent among all $\Delta x_n$. Furthermore, for a linear equilibrium to exist, this condition must hold independently of the realizations of $\tilde{v} - \bar{v}_{n-1}$. Consequently, all coefficients in the first-order condition must vanish:
\begin{align*}
-\frac{2}{\gamma_n^\Sigma}+ \frac{M\gamma_{n+1}\psi_n^2}{2}
=
1-\frac{M\gamma_{n+1}\psi_n}{2}
=
\frac{r^\Sigma_n}{\gamma^\Sigma_n}
+\frac{M\gamma_{n+1}\psi_n^2 s_n}{2}
=0
\end{align*}
subject to the second-order condition $\frac{1}{\gamma^\Sigma_n} - \frac{M\gamma_{n+1}\psi_n^2}{4}>0$.
Rearranging these equations yields
\begin{equation}
\psi_n = \frac{2}{M\gamma_{n+1}} = \frac{2}{\gamma_{n}^\Sigma},
\qquad
\psi_n s_n = -\frac{r^\Sigma_n}{\gamma^\Sigma_n}.
\label{IT_solution}
\end{equation}

{\bf Market maker $m$'s optimization problem}:
Suppose that the informed trader follows the strategy \eqref{disc_informed_strategy} and that each market maker $l\neq m$ adopts the strategy \eqref{disc_mm_strategy}. 
Using the induction hypothesis in \eqref{induction_MM_value}, the market maker's problem reduces to a myopic optimization. 
Substituting \eqref{phi_definition} and \eqref{perceived_price_mm}, we obtain
\begin{align}
&\mathbb{E}\big[(\tilde v- p_n)\Delta y^m_n \,\big|\,\mathcal{F}_{n-1}\vee \sigma(\Delta w_n)\big]\nonumber\\
&=\mathbb{E}\Big[\Big(\tilde v- \bar v_{n-1}-\frac{1}{\gamma^\Sigma_n - \gamma^m_n} \nonumber
    \big(
    \Delta y^m_n +\Delta \hat x_n+\Delta u_n+r^\Sigma_n - r^m_n
    \big)\Big)\Delta y^m_n 
    \,\big|\,\mathcal{F}_{n-1}\vee \sigma(\Delta w_n)\Big]\nonumber\\
&=\Big(d_n- \bar v_{n-1}-\frac{1}{\gamma^\Sigma_n - \gamma^m_n}
    \big(
    \Delta y^m_n +\Delta \hat x_n+\Delta u_n+r^\Sigma_n - r^m_n
    \big)\Big)\Delta y^m_n\nonumber\\
&=\Big(\phi_n(\Delta \hat x_n+\Delta u_n - s_n)
    -\frac{1}{\gamma^\Sigma_n - \gamma^m_n}
    \big(
    \Delta y^m_n +\Delta \hat x_n+\Delta u_n+r^\Sigma_n - r^m_n
    \big)\Big)\Delta y^m_n .\label{value_iterate2}
\end{align}
By the first-order condition, market maker $m$'s optimal order is given by
\begin{equation}\label{MM_optimal}
\Delta \hat y^m_n
= \frac{\phi_n(\gamma^\Sigma_n - \gamma^m_n) -1}{2}(\Delta \hat x_n+\Delta u_n)
-\frac{\phi_n(\gamma^\Sigma_n - \gamma^m_n)}{2}s_n
-\frac{1}{2}(r^\Sigma_n - r^m_n),
\end{equation}
subject to the second-order condition $\frac{1}{\gamma_n^\Sigma - \gamma_n^m}>0$.

{\bf Market clearing}:
Combining the market clearing condition \eqref{market_clearing_condition} with \eqref{IT_solution} and \eqref{MM_optimal} yields
\begin{equation*}
0
=
\Big(\frac{(M-1)\gamma_n^\Sigma\phi_n -M}{2}+1\Big)(\Delta \hat x_n+\Delta u_n)
-\frac{M-1}{2} (\phi_n-\psi_n)\gamma_n^\Sigma s_n.
\end{equation*}
Since this equality must hold for any realization of $\Delta \hat x_n+\Delta u_n$, all coefficients must be zero. Since \eqref{variance_recursion} and \eqref{phi_definition} imply $\phi_n\neq \psi_n$, we obtain
\begin{equation}\label{MM_solution}
\gamma_n^\Sigma =\frac{M-2}{(M-1)\phi_n},
\qquad
s_n = r^\Sigma_n =0.
\end{equation}
Substituting \eqref{perceived_price_mm} and \eqref{MM_solution} into \eqref{disc_mm_strategy}, we obtain
\begin{equation*}
\Delta \hat y_n^m
=
\gamma^m_n(\bar{v}_{n-1} - p_n)+r^m_n
=
-\frac{\gamma_n^m}{\gamma^\Sigma_n - \gamma_n^m}
\big(\Delta \hat y_n^m+\Delta \hat x_n+\Delta u_n - r^m_n \big)
+r^m_n .
\end{equation*}
Solve the above equation with respect to $\Delta \hat y^m_n$ and comparing with \eqref{MM_optimal} and \eqref{MM_solution}, we obtain
\begin{equation}\label{mm_solution_symmetry}
\gamma^m_n
=
\frac{M-2}{M(M-1)\phi_n}
=:\gamma_n,
\qquad
r^m_n =0,
\quad
\text{for } m=1,\dots,M.
\end{equation}

{\bf Verification of the induction step}:
We verify that \eqref{induction_r_k}--\eqref{induction_MM_value} hold for $k=n$. 

\underline{\eqref{induction_r_k} and \eqref{induction_gamma_k} for $k=n$}: These follow directly from \eqref{MM_solution} and  \eqref{mm_solution_symmetry}.

\underline{\eqref{induction_psi_k} for $k=n$}: Substituting \eqref{mm_solution_symmetry} into \eqref{IT_solution}, we obtain
\begin{equation}\label{psi_phi_phi}
    \psi_{n} =\frac{2M-2}{M-2}\phi_{n} = \frac{2M-2}{M-2}\phi_{n+1}.
\end{equation}

\underline{\eqref{induction_Sigma_k} for $k=n$}: Combining \eqref{psi_phi_phi} with \eqref{variance_recursion} and \eqref{phi_definition}, we obtain
\begin{equation}\label{sigma_u_beta_beta}
\sigma_u^2
=
\frac{M}{M-2}\big(\beta_n^2\Sigma_{n-1}\Delta t+\sigma_{z,n}^2\big)
=
\frac{M}{M-2}\big(\beta_{n+1}^2\Sigma_{n}\Delta t+\sigma_{z,n+1}^2\big).
\end{equation}
Substituting \eqref{induction_phi_k} for $k=n+1$, \eqref{psi_phi_phi}, and \eqref{sigma_u_beta_beta} into \eqref{variance_recursion} yields
\[
\Sigma_n
=
\Sigma_{n-1}
-
\frac{4(M-1)^2}{M(M-2)}\phi_{n+1}^2
\sigma_u^2\Delta t =
\Sigma_{n-1}
-
\frac{1}{N-n}\Sigma_n  \quad \Longrightarrow \quad \Sigma_n = \frac{N-n}{N-n+1}\Sigma_{n-1}.
\]

\underline{\eqref{induction_beta_k} for $k=n$}: 
Combining \eqref{psi_phi_phi}, \eqref{phi_definition}, and \eqref{sigma_u_beta_beta}, we obtain 
\[
\beta_n
=
\frac{\Sigma_n}{\Sigma_{n-1}}\beta_{n+1}= \frac{N}{N-n+1}\frac{M-2}{2M-2}\frac{1}{\phi_{n}},
\]
where the last equality is due to \eqref{induction_Sigma_k} for $k=n$, \eqref{induction_beta_k} for $k=n+1$, and $\phi_n=\phi_{n+1}$ from \eqref{psi_phi_phi}.

\underline{\eqref{induction_phi_k} for $k=n$}: This is due to \eqref{induction_Sigma_k} for $k=n$, \eqref{induction_phi_k} for $k=n+1$, and $\phi_n=\phi_{n+1}$.

\underline{\eqref{induction_sigmaz_k} for $k=n$}: 
This is due to \eqref{induction_beta_k}-\eqref{induction_phi_k} for $k=n$ and \eqref{sigma_u_beta_beta}.

\underline{\eqref{induction_IT_value} for $k=n$}: Substituting \eqref{disc_informed_strategy} and \eqref{MM_solution} into \eqref{value_iterate}, we obtain
\begin{align*}
&\mathbb{E}\Big[\sum_{k=n}^N(\tilde v - p_k)\Delta x_k
\;\big|\; \mathcal{F}_{n-1}\vee \sigma(\tilde v)\Big] \nonumber\\
&=\mathbb{E}\Big[ \Big(\tilde v  - \bar v_{n-1}  
- \frac{1}{\gamma^\Sigma_n}\Delta \hat x_n\Big)\Delta \hat x_n 
+ \frac{M\gamma_{n+1}}{4}\big(\tilde v- \bar{v}_{n-1} - \psi_n\Delta \hat x_n)\big)^2 \;\big|\; \mathcal{F}_{n-1}\vee \sigma(\tilde v)\Big] \\
&=\Big( \beta_n \Delta t \Big( 1- \frac{\beta_n \Delta t}{\gamma_n^\Sigma}  \Big) + \frac{M\gamma_{n+1}}{4} (1-\psi_n \beta_n\Delta t)^2 \Big) (\tv - \bv_{n-1})^2 + \Big(  - \frac{1}{\gamma_n^\Sigma}  + \frac{M \gamma_{n+1} \psi_n^2}{4}\Big) \sigma_{z,n}^2 \Delta t\\
&=\frac{M \gamma_n}{4}(\tv - \bv_{n-1})^2,
\end{align*}
where the last equality is obtained by combining \eqref{induction_gamma_k} for $k=n,n+1$, \eqref{induction_psi_k} for $k=n$, and \eqref{MM_solution}.

\underline{\eqref{induction_MM_value} for $k=n$}: We obtain $\Delta \hat y^m_n=-\frac{1}{M}(\Delta \hat x_n + \Delta u_n)$ from \eqref{MM_optimal}-\eqref{mm_solution_symmetry}. Then, \eqref{induction_MM_value} for $k=n+1$, \eqref{disc_informed_strategy}, \eqref{value_iterate2}, \eqref{MM_solution}, and \eqref{mm_solution_symmetry} yield
\begin{align*}
\mathbb{E}\Big[\sum_{j=n}^N(\tilde v - p_j)\Delta \hat y^m_j 
\Big| \mathcal{F}_{n-1}\Big] &= \mathbb{E}\big[(\tilde v - p_n)\Delta \hat y^m_n 
\big| \mathcal{F}_{n-1}\big]+ \nu_n\\
&=\frac{1}{M^2(M-1)\gamma_n} \, \mathbb{E}\big[  (\Delta \hat x_n+ \Delta u_n)^2 \big| \mathcal{F}_{n-1}\big] + \nu_n\\
&=\frac{\beta_n \Sigma_{n-1} \Delta t}{M^2(M-1)\gamma_n \phi_n} + \nu_n= \frac{\beta_n \Sigma_{n-1} \Delta t}{M(M-2)}+ \nu_n,
\end{align*}
where we use \eqref{phi_definition} for the third equality.

To complete the proof, we observe that \eqref{induction_Sigma_k} and \eqref{induction_phi_k}, together with $\Sigma_0=\sigma_v^2$, imply
\begin{align}
\phi_n  =\sqrt{\frac{M(M-2)}{(2M-2)^2}\frac{\sigma_v^2}{\sigma_u^2}}\quad \textrm{and} \quad\Sigma_n = \frac{N-n}{N} \sigma_v^2 \quad \textrm{for $n=1,2,...,N$.} \label{phi_Simga_final}
\end{align}
We substitute these into \eqref{induction_gamma_k}-\eqref{induction_beta_k} and \eqref{induction_IT_value}-\eqref{induction_MM_value} to obtain the expressions in the theorem.

\bigskip

\section{Price Autocorrelation}\label{app:price_autocorr}
In this appendix, we compute the autocorrelation of equilibrium price changes reported in Table~\ref{tab:model_comparison}. 
Unless stated otherwise, all equilibrium relationships used in the derivations are taken from Theorem~\ref{thm:equilibrium}.

Since the process $\{\bar v_n\}_{n=0}^N$ is an $(\mathcal F_n)_{n=0}^N$-martingale, it follows that
\begin{equation}
\begin{split}\label{bv_mart}
    \mathbb{E}[\Delta \bv_n \Delta \bv_m] &= 0 \quad \textrm{for}\quad 1\leq n \neq m \leq N,\\
    \mathbb{E}[(\Delta \bv_n)^2] &= \mathbb{E}[(\tv-\bv_{n-1})^2- (\tv-\bv_{n})^2- 2\Delta \bv_n (\tv-\bv_{n})] = \Sigma_{n-1}-\Sigma_{n} \\
    &= \sigma_v^2 \Delta t \quad \textrm{for}\quad 1\leq n\leq N,
\end{split}
\end{equation}
where the last equality follows from \eqref{phi_Simga_final}. 

For fixed $1\leq n, m \leq N$, suppose that $\mathbb{E}[\bv_{n-1} \Delta u_m]=0$. Then, 
\begin{align*}
\mathbb{E}[\bv_{n} \Delta u_m]&=\mathbb{E}[\big(\bv_{n-1} + \beta_{n}(\tv - \bv_{n-1})+\Delta z_{n}\big) \Delta u_m] =0,
\end{align*}
where the last equality follows from the independence of $\Delta u_m$ from $\tv$ and $\Delta z_{n}$. Since $\bv_0=0$, mathematical induction ensures that $\mathbb{E}[\bv_{n} \Delta u_m]=0$ for all $n,m$. In particular, this implies
\begin{align}
\mathbb{E}[\Delta \bv_{n} \Delta u_m]&=0 \quad \textrm{for}\quad 1\leq n, m \leq N. 
\end{align}

We substitute the result from \thmref{thm:equilibrium} into \eqref{perceived_price_informed} to obtain the equilibrium price:
\begin{align}
p_n &= \bar v_{n-1}+\frac{1}{M\gamma_n}\Delta w_n = \bar v_{n-1}+\frac{\psi_n}{2}(\Delta x_n +\Delta u_n)= \bar v_{n-1}+\frac{1}{2}\Delta \bar v_n +\frac{\psi_n}{2}\Delta u_n \nonumber\\
&= \frac{1}{2}\big( \bar v_{n-1}+\bar v_n +\psi \Delta u_n \big),
\quad \textrm{where}\quad \psi := \psi_n= \sqrt{\tfrac{M}{M-2}}\tfrac{\sigma_v}{\sigma_u} \textrm{   (no $n$-dependence).}
 \label{equilibrium_p}
\end{align}
Combining \eqref{bv_mart}-\eqref{equilibrium_p}, we obtain
\begin{align*}
\mathbb{E}[(\Delta p_n)^2] &=\frac{1}{4}\Big( \mathbb{E}[(\Delta \bv_{n-1})^2] + \mathbb{E}[(\Delta \bv_{n})^2] + 2\psi^2 \sigma_u^2 \Delta t   \Big) = \frac{M-1}{M-2} \sigma_v^2 \Delta t,\\
\mathbb{E}[\Delta p_n \Delta p_{n+1}]&=\frac{1}{4}\Big(  \mathbb{E}[(\Delta \bv_{n})^2] -\psi^2 \sigma_u^2 \Delta t   \Big) = - \frac{1}{2(M-2)} \sigma_v^2 \Delta t.
\end{align*}
Since all random variables have mean zero, the autocorrelation is given by
\begin{align*}
   \mathrm{Corr}(\Delta p_n, \Delta p_{n+1})
    = \frac{\mathrm{Cov}(\Delta p_n, \Delta p_{n+1})}{\sqrt{\mathrm{Var}(\Delta p_n)\mathrm{Var}(\Delta p_{n+1})}}
    = -\frac{1}{2(M-1)}.
\end{align*}

\section{Cross-Country Evidence on Disclosure and Competition}\label{app:cross_country}

In Section~\ref{sec:empirical}, we test our central prediction using CRSP market maker counts as a direct measure of competition among U.S.\ stocks. In this appendix, we provide complementary evidence by reinterpreting cross-country findings from two existing studies within our theoretical framework. Because these studies were not designed to test our model, the proxy mappings described below are necessarily indirect; we discuss their limitations alongside the analysis.

\subsection*{Evidence from \cite{FGH2006}}

\cite{FGH2006} examine cross-sectional associations between stock exchange disclosure systems and market development across 50 international exchanges. They document a positive association between the strength of disclosure and market quality, showing that exchanges with stronger disclosure tend to exhibit higher levels of market development and trading activity. To refine this observation within our theoretical framework, we utilize three key variables reported in Table 2 of \cite{FGH2006}. First, we take $\mathit{MinorityCap}$ from Panel A, defined as the \emph{stock market capitalization held by minorities deflated by GDP}, as a proxy for the degree of market-making competition $M$. This choice rests on the economic intuition that larger, more developed financial markets with broader investor participation foster a more competitive environment for financial intermediaries, supporting a greater number of active market makers.
We acknowledge that $\mathit{MinorityCap}$ is a broad measure of market development rather than a direct count of market makers. In principle, the correlation between $\mathit{MinorityCap}$ and market quality could also reflect differences in regulatory environments, investor sophistication, or other institutional features that covary with market size. Nevertheless, to the extent that more developed markets support a larger and more competitive intermediary sector, the mapping to our competition parameter $M$ is economically plausible.

Several considerations support the use of $\mathit{MinorityCap}$ as a proxy for intermediary competition. First, among the five components of the market development index in \cite{FGH2006}, $\mathit{MinorityCap}$ is the variable most closely tied to investor participation breadth: it measures the equity stake available to dispersed outside investors, which is a necessary condition for sustaining a large and active intermediary sector. Markets in which minority shareholders hold a substantial share of capitalization are precisely those in which trading demand is dispersed across many participants, creating profit opportunities that attract competing market makers. Second, this variable directly descends from the minority shareholder protection measures of \cite{LLSV1998}, which have been widely used in the law-and-finance literature to capture the breadth of capital-market participation. \cite{LLSV1998} show that countries with stronger legal protection of minority investors have larger and more liquid equity markets, an environment that naturally supports a more competitive financial intermediary sector.

To evaluate how the impact of disclosure varies with this competition proxy, we define a disclosure-adjusted measure of trading activity for each exchange as
\[
\mathit{TradeActivity}^{D} := \frac{\mathit{Trans}}{D+1},
\]
where $\mathit{Trans}$ denotes the \emph{annual number of transactions in equity shares deflated by year-end market capitalization} from Table 2, Panel A of \cite{FGH2006} and $D$ denotes the \emph{overall disclosure} measure from Table 2, Panel B of \cite{FGH2006}.\footnote{We restrict attention to the 41 exchanges for which both $\mathit{Trans}$ and $D$ are jointly available, representing the maximum available sample. The term $D+1$ ensures a strictly positive denominator, as the standardized measure $D$ can take negative values with a minimum of $-0.68$.} This ratio, $\mathit{TradeActivity}^D$, serves as a conceptual measure of ``liquidity efficiency,'' representing the amount of trading activity normalized by disclosure intensity.

\renewcommand{\thefigure}{C.1}
\begin{figure}[t]
\centering
\includegraphics[width=0.5\linewidth]{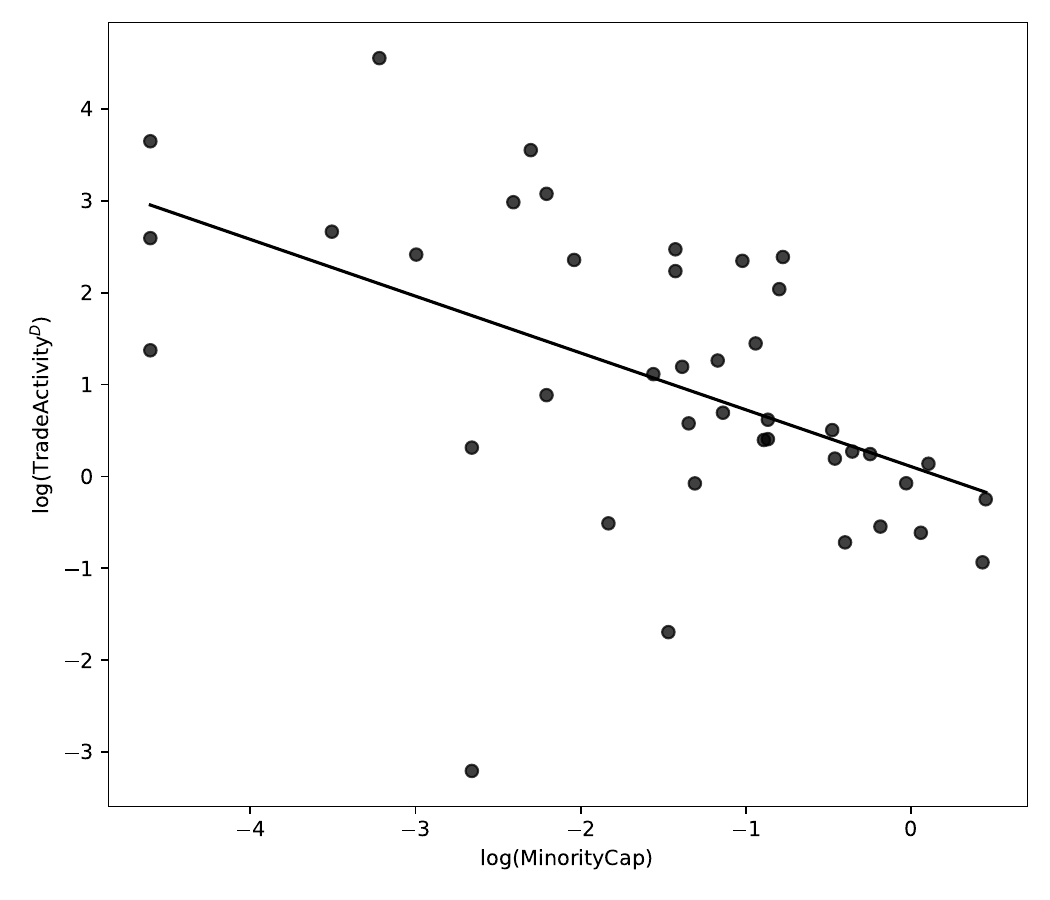}
\caption{
Log--log plot of disclosure-adjusted trading activity ($\mathit{TradeActivity}^D$) against market competition proxy ($\mathit{MinorityCap}$). The fitted line exhibits a negative relationship (slope = $-0.62$, correlation = $-0.52$), which suggests that the liquidity-enhancing effect of disclosure is more pronounced in less competitive markets. Data are based on a sample of 41 international exchanges from \cite{FGH2006}.
}
\label{fig:disclosure_trading}
\end{figure}

Figure~\ref{fig:disclosure_trading} presents the $\log$--$\log$ relationship between $\mathit{TradeActivity}^D$ and $\mathit{MinorityCap}$. The fitted line exhibits a negative slope of $-0.62$ and a correlation of $-0.52$. Under our interpretation, this negative cross-sectional association is consistent with the predictions in \proref{cor:disclosure_impact}. Specifically, the higher values of $\mathit{TradeActivity}^D$ observed in smaller, less developed markets (low $\mathit{MinorityCap}$) suggest that disclosure provides a more substantial reduction in market frictions when competition among market makers is weak. As competition intensifies (high $M$), this marginal benefit diminishes, which is reflected in the lower levels of disclosure efficiency observed across more developed exchanges.

\subsection*{Evidence from \cite{LLM2012}}

\cite{LLM2012} examine the relation between firm-level transparency and stock market liquidity in an international setting, focusing on how this relation varies across institutional environments. Their key empirical results are summarized in Table 5, Panel B of their study, which examines how the sensitivity of aggregate stock market illiquidity to firm-level transparency varies across country-level governance structures. We report the relevant coefficients from this panel in Table~\ref{tab:LLM_panelB}.
In Table 5, Panel B of \cite{LLM2012}, market quality is captured by the variable $\mathit{ILLIQ}$, which is an aggregate measure of illiquidity constructed from bid--ask spreads and the proportion of zero-return trading days. Firm-level transparency is measured by $\mathit{TRANS}$, a composite index summarizing multiple dimensions of financial reporting quality. Country-level institutional environments are summarized by $\mathit{GOV\_SCORE}$, a governance measure reported in Table 1 of \cite{LLM2012} that incorporates factors such as investor protection and media penetration. This score classifies countries into four groups, ranging from 0 for weak governance to 3 for strong governance.

\renewcommand{\thetable}{C.1}
\begin{table}[t]
\centering
\caption{Estimated coefficients of $\mathit{TRANS}$ on $\mathit{ILLIQ}$ across governance groups}
\label{tab:LLM_panelB}
\begin{center}
\begin{tabular}{lcccc}
\toprule
 & \multicolumn{4}{c}{$\mathit{GOV\_SCORE}$} \\
\cmidrule(lr){2-5}
Variable
& (0)
& (1)
& (2)
& (3) \\
\midrule
$\mathit{TRANS}$
& $-0.317$
& $-0.194$
& $-0.125$
& $-0.117$ \\
\bottomrule
\end{tabular}
\end{center}
\vspace{0.5em}
\parbox{\textwidth}
{    \scriptsize
    \textit{Note}: This table reports the estimated coefficients of firm-level transparency ($\mathit{TRANS}$) on aggregate illiquidity ($\mathit{ILLIQ}$) from Table 5, Panel B of \cite{LLM2012}. $\mathit{GOV\_SCORE}$ ranges from 0 (weak governance) to 3 (strong governance).}
\end{table}

We interpret $\mathit{GOV\_SCORE}$ as an empirical proxy for the degree of effective market-making competition $M$. This interpretation rests on the economic intuition that stronger institutional environments are associated with lower informational rents and fewer barriers to entry for financial intermediaries, thereby fostering a more competitive market-making sector. Conversely, weak governance environments, characterized by low values of $\mathit{GOV\_SCORE}$, are more likely to reflect settings with limited competition among market makers.
As with $\mathit{MinorityCap}$, this mapping is indirect. Governance quality simultaneously captures investor protection, legal enforcement, and media scrutiny, all of which may independently affect liquidity through channels other than market-maker competition. Our interpretation therefore isolates one plausible mechanism among several. The monotonic pattern documented below is consistent with our model but does not rule out alternative explanations.

Nonetheless, several features of $\mathit{GOV\_SCORE}$ make it a reasonable proxy for our purposes. First, $\mathit{ASDI}$, the anti-self-dealing index from \cite{DPLS2008}, captures the legal protection of outside investors against expropriation by insiders. In countries where such protection is weak, concentrated ownership and insider control are prevalent, reducing the scope for arm's-length trading and thereby limiting the competitive market-making sector. Second, $\mathit{MEDIA}$, which measures print and television penetration, proxies for the richness of the information environment available to market participants. Richer information environments lower the informational advantage of better-informed traders, which in turn reduces the adverse-selection costs facing market makers and facilitates entry. Taken together, these components of $\mathit{GOV\_SCORE}$ are economically linked---through complementary channels---to the conditions that sustain a competitive intermediary sector. A similar logic underlies the subsample analysis in \cite{LLM2012} themselves, who note that firm-level transparency is a substitute for weak country-level institutions, consistent with our prediction that disclosure matters most precisely when the competitive discipline among intermediaries is weakest.

Under this mapping, the empirical evidence in Table~\ref{tab:LLM_panelB} is qualitatively consistent with our theoretical prediction. The regression coefficients represent the marginal impact of transparency on illiquidity; notably, this effect is strongest in magnitude for the low-$\mathit{GOV\_SCORE}$ group ($-0.317$) and monotonically diminishes as governance improves (reaching $-0.117$ for Group 3). These results support the implication of \proref{cor:disclosure_impact} that the marginal liquidity benefit of disclosure is greater when market-making competition is weaker.



\end{document}